\documentclass{llncs}

\usepackage[utf8]{inputenc}
\usepackage{amsmath, amssymb, graphicx, url}
\usepackage[english]{babel}
\usepackage{todonotes}
\usepackage{paralist}
\usepackage{caption}
\usepackage{subcaption}
\usepackage{times}
\usepackage{microtype}

\usepackage[raiselinks=true, bookmarks=true,
	bookmarksopenlevel=1,
	bookmarksopen=true,
	bookmarksnumbered=true,
	hyperindex=true,
	plainpages=false,
	pdfpagelabels=true,
	pdfborder={0 0 0.5},
	colorlinks=false,						
	linkbordercolor={0 0.61 0.50},   
	citebordercolor={0 0.61 0.50}]{hyperref}

\usepackage{xcolor}

\let\doendproof\endproof
\renewcommand{\endproof}{\qed\doendproof}
 
\newcommand{\myparagraph}[1]{\smallskip\noindent\textit{#1} }
\newcommand{\boldparagraph}[1]{\medskip\noindent\textbf{#1} }
 
\title{Graphs with Plane Outside-Obstacle
  Representations
}

\author{Alexander Koch \and Marcus Krug \and Ignaz Rutter}

\institute{Karlsruhe Institute of Technology (KIT)\\
  \texttt{alexander.koch2@student.kit.edu, krug.marcus@gmail.com,
    rutter@kit.edu}}
 
\graphicspath{{./fig/}}
  
\begin{document}

\pagestyle{plain}
\maketitle

\begin{abstract}
  An \emph{obstacle representation} of a graph consists of a set of
  polygonal obstacles and a distinct point for each vertex such that
   two points see each other if and only if the corresponding
  vertices are adjacent.  Obstacle representations are a recent
  generalization of classical polygon--vertex visibility graphs, for
  which the characterization and recognition problems are
  long-standing open questions.

  In this paper, we study \emph{plane outside-obstacle
    representations}, where all obstacles lie in the unbounded face of
  the representation and no two visibility segments cross.  We
  give a combinatorial characterization of the biconnected graphs that
  admit such a representation.  Based on this characterization, we
  present a simple linear-time recognition algorithm for these graphs.
  As a side result, we show that the plane vertex--polygon visibility
  graphs are exactly the maximal outerplanar graphs and that every
  chordal outerplanar graph has an outside-obstacle representation.
\end{abstract}

\section{Introduction}

Visibility, and hence, visibility representations of graphs are
central to many areas, such as architecture, sensor networks, robot
motion planning, and surveillance and security.  There is a long
history of research on characterizing and recognizing visibility
graphs in various settings; see the related work below.  Despite
tremendous efforts characterizations and efficient recognition
algorithms are only known for very restricted
cases~\cite{ec-rvgsp-90,ec-nrcvg-95}.  Recently, Alpert et
al.~\cite{akl-ong-10} introduced obstacle representations of graphs,
which generalize many previous visibility variants, such as
polygon--vertex visibility.  In this paper, we study \emph{plane
  outside-obstacle representations}, where the visibility segments may
not cross, and a single obstacle is located in the outer face of the
representation.  We characterize the biconnected graphs admitting such
a representation and give a linear-time recognition algorithm.  This
is one of the first results that characterizes such a class of graphs
and gives an efficient recognition algorithm.  In the following we
first give some basic definitions.  Afterwards, we present an overview
of related work and describe our contribution in more detail.

An \emph{obstacle representation} of a graph~$G=(V,E)$ consists of a
set of polygonal obstacles and a distinct point for each vertex
in~$V$.  The representation is such that two points see each other
if and only if the corresponding vertices are adjacent.  The
\emph{obstacle number} of~$G$ is the smallest number of obstacles in
any obstacle representation of~$G$.

In an \emph{outside-obstacle representation} all obstacles are in the
unbounded face of the representation, i.e., they are contained in the
unbounded face of the corresponding straight-line drawing of the
graph.  Outside-obstacle representations are a recent generalization
of classical \emph{polygon--vertex visibility} graphs, where the
obstacle is a simple polygon, the points are the vertices of the
polygon and visibility segments have to lie inside the polygon.  The
corresponding characterization problem and the complexity of the recognition
problem are long-standing open questions.

Figure~\ref{fig:examples-oor} shows examples of outside-obstacle
representations.  We note that, in the case of outside-obstacle
representations, it can be assumed that there exists only a single
obstacle that surrounds the outer face of the representation, as
illustrated in Fig.~\ref{fig:example-oor-1}.  In particular, graphs
with an outside-obstacle representation have obstacle number at
most~1.  In our setting, given a straight-line drawing of a graph, we
hence do not make the outside obstacle explicit.  Instead, we require
for an outside-obstacle representation of a graph~$G$ that any
non-edge of~$G$ intersects the outer face of~$G$.  It is then not
difficult to construct a corresponding obstacle.  Throughout this
document, we assume that points and vertices of obstacles are in
\emph{general position}, so no three of them are collinear.

\begin{figure}[tb]
  \centering  
  \begin{subfigure}[b]{.45\textwidth}
    \centering
    \includegraphics[page=5]{fig/examples-oor}
    \caption{}\label{fig:example-oor-1}
  \end{subfigure}\hfil
  \begin{subfigure}[b]{.45\textwidth}
    \centering
    \includegraphics[page=1]{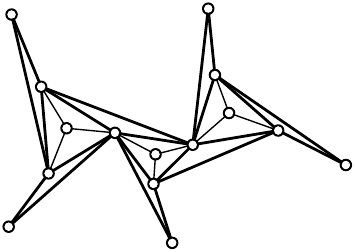}
    \caption{}\label{fig:example-oor-2}
  \end{subfigure}
  \caption{(a) A non-planar outside-obstacle representation of the
    octahedron; the obstacle is shown in gray.  (b) A plane outside-obstacle representation of
    Figure~\ref{fig:example-3}.}
  \vspace{-3ex}
  \label{fig:examples-oor}
\end{figure}

\myparagraph{Related Work.}
Graphs with an outside-obstacle representation are equivalent to
visibility graphs of a pointset within a simple polygon. Therefore,
polygon--vertex visibility graphs form an important subclass of our
graphs class, where the pointset coincides with the corners of the
surrounding simple polygon.  Such graphs have been extensively studied
due to their many applications, e.g., in gallery
guarding~\cite{r-agta-87}.

Polygon--vertex visibility graphs were first introduced in 1983 by
Avis and ElGindy~\cite{ae-caps-83} and are most studied in the field
of visibility problems~\cite{gg-upvgpsp-12}.  One of the first results
on the topic was that maximal outerplanar graphs are polygon--vertex
visibility graphs~\cite{e-hdpa-85}.  Ghosh~\cite{g-rcvgs-97} gives a
set of four necessary conditions for polygon--vertex visibility
graphs, which he conjectured to be also sufficient. However,
Streinu~\cite{s-npvg-05} constructed a counterexample.  As pointed out
in~\cite{gg-upvgpsp-12}, two of Ghosh's necessary
conditions~\cite{g-rcvgs-97} imply the conditions of a
characterization attempt by Abello and Kumar in terms of oriented
matroids~\cite{ak-vgom-02}, which hence cannot be sufficient.

So far, characterizations have only been achieved for polygon--vertex
visibility graphs of restricted polygons.  Everett and
Corneil~\cite{ec-rvgsp-90,ec-nrcvg-95} give a characterization of
visibility graphs in spiral and 2-spiral polygons -- polygons that
have exactly one and two chains of reflex vertices, respectively.
They are characterized as interval graphs and perfect graphs.  The
best known complexity result about the recognition and reconstruction
problem of polygon--vertex visibility graphs is that they are in
PSPACE~\cite{e-vgr-90}.

A generalization of these graphs are induced subgraphs of
polygon--vertex visibility graphs, or \emph{induced visibility graphs}
for short.  Spinrad~\cite{s-egr-03} considers this graph class the
natural generalization of polygon--vertex visibility graphs, which is
hereditary with respect to induced subgraphs.
Everett and Corneil~\cite{ec-nrcvg-95} show that there is no finite
set of forbidden induced subgraphs in polygon--vertex visibility
graphs.


Coullard and Lubiw show that 3-connected polygon--vertex visibility
graphs admit a 3-clique ordering~\cite{cl-dvg-91}.
Abello et al.~\cite{alp-vgsp-92} prove that every 3-connected planar
polygon--vertex visibility graph is maximal planar and that every
4-connected such graph cannot be planar.  Their conjecture that
Hamiltonian maximal planar graphs with a 3-clique ordering are
polygon--vertex visibility graphs was disproven by Chen and
Wu~\cite{cw-dcpvg-01}.  According to O'Rourke~\cite{o-opcvi-98}
necessary and sufficient conditions for a polygon--vertex visibility
graph to be planar are known~\cite{lc-pvg-94}, but do not lead to a
polynomial recognition algorithm.


For the more general question of obstacle numbers, Alpert et
al.~\cite{akl-ong-10} give a construction for graphs with large
obstacle number and small example graphs that have obstacle number
greater than~1.  They further show that every outerplanar graph admits
a (non-planar) outside-obstacle representation, i.e., they are
visibility graphs of pointsets inside simple polygons.  Subsequent
papers extend the results on obstacle numbers.  Pach and
Sarıöz~\cite{ps-sglon-11} construct small graphs with obstacle
number~2 and show that bipartite graphs with arbitrarily large
obstacle number exist.  Mukkamala et al.~\cite{mpp-lbong-12} show that
there are graphs on $n$ vertices with obstacle number
$\Omega(n/\log{n})$. It is an open question whether any graph with
obstacle number~1 admits an outside-obstacle representation.

\myparagraph{Contribution and Outline.}
\begin{figure}[tb]
  \centering  
  \begin{subfigure}[b]{.1\textwidth}
    \centering
    \includegraphics[page=1]{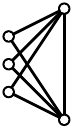}
    \caption{}\label{fig:example-1}
  \end{subfigure}\hfil
  \begin{subfigure}[b]{.3\textwidth}
    \centering
    \includegraphics[page=2]{fig/examples}
    \caption{}\label{fig:example-2}
  \end{subfigure}\hfil
  \begin{subfigure}[b]{.3\textwidth}
    \centering
    \includegraphics[page=3]{fig/examples}
    \caption{}\label{fig:example-3}
  \end{subfigure}
  \caption{Graphs admitting (c) / not admitting (a,b) a
    plane outside-obstacle representation.}  \vspace{-3ex}
  \label{fig:examples}
\end{figure}
In this paper, we study \emph{plane outside-obstacle representations},
where the drawing of~$G$, without the obstacles, is free of crossings;
see Fig.~\ref{fig:example-oor-2} for an example.  Consider the graphs
shown in Fig.~\ref{fig:examples}.  We will see that the first two
graphs do not admit a plane outside-obstacle representation, whereas
the last example has one.  Note that the drawing in
Fig.~\ref{fig:example-1} is a (non-planar) outside-obstacle
representation.  Our main results are the following.
\begin{inparaenum}[(1)]
\item Every outerplanar graph whose inner faces are triangles admits a
  plane outside-obstacle representation.
\item A characterization of the biconnected graphs that admit a plane
  outside-obstacle representation.
\item A linear-time algorithm for testing whether a biconnected graph
  admits a plane outside-obstacle representation.
\end{inparaenum}
As a side result, we obtain a simple combinatorial proof of ElGindy's
classical result that maximal outerplanar graphs are polygon--vertex
visibility graphs~\cite{e-hdpa-85}.

Our paper is structured as follows.  First, we derive a simple
necessary condition on the structure of biconnected graphs that admit
a plane outside-obstacle representation in
Section~\ref{sec:inner-chordal}.  This restricts the class of graphs
we have to consider and we derive some useful structural results about
such graphs.  Afterwards, in Section~\ref{sec:characterization}, we
give a local description of plane outside-obstacle representations
and, based on this, we derive a combinatorial characterization of the
biconnected planar graphs that admit a plane outside-obstacle
representation in terms of an orientation of a certain subset of
edges.  Using this characterization, we prove our main results in
Section~\ref{sec:main}.
%

\section{Inner-Chordal Plane Graphs}
\label{sec:inner-chordal}

A graph with a fixed planar embedding is \emph{inner-chordal plane} if
any cycle~$C$ of length at least~4 has a chord that is embedded in the
bounded region of~$\mathbb{R}^2 \setminus C$; see
Fig.~\ref{fig:inner-chordal}.  We first show that we can restrict our
analysis to inner-chordal plane graphs.

\begin{lemma}\label{lem:plane-oor-chordal}
  Graphs with a plane outside-obstacle representation are
  inner-chordal plane.
\end{lemma}
\begin{proof}
  Let~$G$ be a graph with a plane outside-obstacle representation, and
  assume it is not inner-chordal.  Hence, there exists a cycle~$C$ of
  length at least~4, whose interior does not contain a chord.  Note
  that the obstacle lies outside of~$C$ by definition.  The cycle~$C$
  is embedded as the boundary of a simple polygon~$P$ on at least four
  vertices.  Since~$P$ can be triangulated, there exists a pair of
  non-adjacent vertices~$u$ and~$v$ on~$C$ such that the segment~$uv$
  is completely contained in~$P$.  Hence the obstacle cannot
  intersect~$uv$, and thus~$\{u,v\} \in E(G)$ by definition,
  contradicting our choice of~$u$ and~$v$.
\end{proof}

\begin{figure}[tb]
  \centering
  \begin{subfigure}[b]{.2\textwidth}
    \includegraphics[page=1]{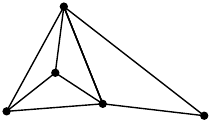}
      \caption{}   
  \end{subfigure}\hfil
  \begin{subfigure}[b]{.2\textwidth}
    \centering
    \includegraphics[page=3]{inner-chordal}
    \caption{}
  \end{subfigure}\hfil
  \begin{subfigure}[b]{.2\textwidth}
    \centering
    \includegraphics[page=2]{inner-chordal}
    \caption{}
  \end{subfigure}
  \caption{(a) An inner-chordal graph. (b), (c) Chordal, but not inner-chordal, plane graphs.}
  \vspace{-2ex}
  \label{fig:inner-chordal}
\end{figure}

Note that Lemma~\ref{lem:plane-oor-chordal} shows immediately that the
graph from Fig.~\ref{fig:example-1} does not admit a plane outside-obstacle
representation.  Although this graph is chordal, it does not have a
planar embedding that is inner-chordal.
In the following, we consider only inner-chordal plane graphs.  Note
that, in particular, every inner face of an inner-chordal plane graph
is a triangle.  Moreover, an outerplanar graph is inner-chordal if and
only if it is chordal, which is the case if and only if every inner
face is a triangle.

\begin{lemma}
  \label{lem:inner-chordal-degree}
  Let~$G$ be an inner-chordal plane graph.  Then, every inner vertex
  of~$G$ has degree~3 and no two inner vertices are adjacent.
\end{lemma}

\begin{proof}
  Since~$G$ is inner-chordal, every inner face is necessarily a
  triangle.  This implies that the neighbors of any inner vertex form
  a cycle.  This cycle does not have an inner chord, and hence,
  since~$G$ is simple and inner-chordal, it must have length~3.  This
  implies that any inner vertex has degree~3.  Moreover, it follows
  that the neighbors of any inner vertex~$v$ must be incident to the
  outer face as they already have degree~3 in the neighborhood of~$v$.
\end{proof}

This description also gives rise to a certain tree that is associated
with every biconnected inner-chordal graph.  Let~$G$ be a biconnected
inner-chordal plane graph.  A \emph{chord} of~$G$ is an edge that is
not incident to the outer face but whose endpoints are incident to the
outer face.  Lemma~\ref{lem:inner-chordal-degree} implies a
decomposition of~$G$ along its chords.  Namely, we first remove all
inner vertices.  Each such removal transforms a 4-clique of~$G$ into a
triangular face; we mark each triangle that results from such a
removal.  The resulting graph~$G'$ is outerplanar and every inner face
is a triangle.  Now the weak dual~$T'$ of~$G'$ is a tree, where each
node corresponds to a triangle of~$G'$.  By marking the nodes of~$T'$
that correspond to a marked triangle, we obtain the construction
tree~$T$ of~$G$, denoted~$T(G)$.  Note that each marked node of~$T(G)$
corresponds to a 4-clique of~$G$, whereas an unmarked node corresponds
to a triangular face of~$G$.  We refer to these nodes as~$K_4$-
and~$K_3$-nodes, respectively.  The edges of~$T(G)$ correspond
bijectively to the chords of~$G$.  We refer to the vertices of~$T(G)$
as nodes to distinguish them from the vertices of~$G$.  For a
node~$\tau$ of~$T(G)$, we denote by~$V_\tau$ the vertices of the
corresponding triangle or 4-clique.  

\begin{figure}[tb]
  \centering
  \includegraphics{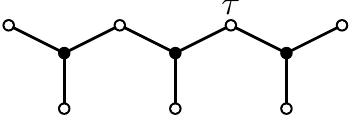}
  \caption{The construction tree of graphs~(b) and~(c) in
    Fig.~\ref{fig:examples}; $K_3$-nodes are empty and $K_4$-nodes are
    filled.  Structurally, the only difference between the two graphs
    is how the two subgraphs of~$G$ are attached to the triangle
    corresponding to the $K_3$-node~$\tau$.}
  \vspace{-3ex}
\end{figure}

Observe that, if we store with each node of~$T(G)$ the corresponding
edges and use the bijection of the edges of~$T(G)$ with the chords
of~$G$ to find the shared chord of adjacent nodes, we can reobtain~$G$
from~$T(G)$ by merging triangles and 4-cliques that are adjacent
in~$T(G)$ along their shared chords.  Then~$T(G)$ is essentially the
SPQR-tree of~$G$~\cite{dt-omtc-96}.  We decided to avoid the
technical machinery associated with SPQR-trees and rather work with
the construction tree, which is more tailored to our needs.

\section{Characterization of Plane Visibility Representations}
\label{sec:characterization}

In this section we devise a combinatorial characterization of the
biconnected inner-chordal graphs that admit a plane outside-obstacle
representation.  This is done in two steps.  First, we show that,
aside from being free of crossings, the property of being a plane
outside-obstacle representation depends only on local features in the
drawing, namely, for each chord of a graph~$G$, its neighbors must be
embedded in certain regions.  In a second step, we show that this
essentially induces a binary choice for each chord.  In this way, an
outside-obstacle orientation induces an orientation of the chords
of~$G$, and we will characterize the existence of a plane
outside-obstacle representation in terms of existence of a suitable
chord orientation.

\boldparagraph{Local Description of Plane Visibility Representations.}
Next we aim to understand better which planar straight-line drawings
are outside-obstacle representations.  As a first observation,
consider two triangles~$D$ and~$D'$ sharing a common edge~$e=\{u,v\}$,
which then forms a chord.  Let~$w$ and~$w'$ denote the tips of~$D$
and~$D'$ with respect to base~$e$, respectively.  For an
outside-obstacle representation it is a necessary condition that the
non-edge~$\{w,w'\}$ intersects the outer face.  We thus have to
position the tips in such a way that the segments~$ww'$ does not lie
inside the drawing of~$D$ and~$D'$.  We use the following definition;
see Fig.~\ref{fig:region} for an illustration.

\begin{definition}\label{dfn:regions1}
  Let $D$ be a triangle and $u$ a vertex of $D$.  Then $R_D(u)$
  denotes the intersection of the half-planes defined by the sides
  of~$D$ incident to~$u$ not containing~$D$.
\end{definition}

\begin{figure}[tb]
  \centering  
  \begin{subfigure}[b]{.5\textwidth}
    \centering
    \includegraphics{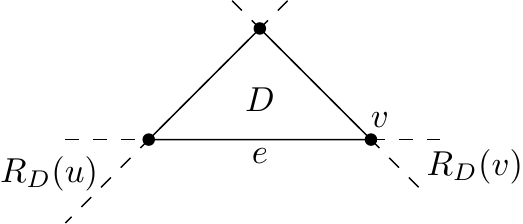}
    \caption{}\label{fig:region}
  \end{subfigure}\hfil
\begin{subfigure}[b]{.45\textwidth}
  \centering
    \includegraphics[page=1]{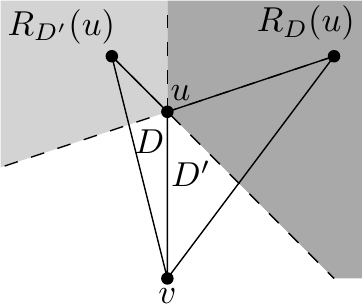}
    \caption{}\label{fig:attachment}
  \end{subfigure}
  \caption{(a) Regions of triangle~$D$ at edge~$e=\{u,v\}$ and (b)
    attachment of two triangles~$D$ and~$D'$ along chord~$\{u,v\}$.}
  \label{fig:proof-regions2}
  \vspace{-2ex}
\end{figure}

To ensure that the segment~$ww'$ intersects the outer face, it is
clearly necessary that~$w' \in R_D(u) \cup R_D(v)$ or~$w \in R_{D'}(u)
\cup R_{D'}(v)$.  The former ensures that the segment~$ww'$ does not
intersect the interior of~$D$ and the letter ensures the same property
for~$D'$.  These intersections behaviors are not independent.  It is
in fact readily seen that~$w' \in R_D(x)$ if and only if~$w \in
R_{D'}(x)$ for~$x \in \{u,v\}$; see Fig.~\ref{fig:attachment}.  More generally, the
same observations also hold for a chord~$e= \{u,v\}$ that is shared by
\begin{inparaenum}[(a)]
  (a) two triangles~$\tau$ and~$\tau'$, (b) a triangle~$\tau$ and a
  4-clique~$\tau'$ and (c) by two 4-cliques~$\tau$ and~$\tau'$.
\end{inparaenum}
To see this, note that the regions~$R_{D}(u)$ and~$R_{D'}(u)$ in
Fig.~\ref{fig:attachment} do not change if~$D$ and/or $D'$ are part of
a 4-clique.  More formally, let~$W$ and~$W'$ be the vertices of~$\tau$
and~$\tau'$ that are distinct from~$u$ and~$v$, respectively.  Let~$D$
and~$D'$ denote the triangles incident to~$e$.  Then the following
condition is necessary
\vspace{-1ex}
\begin{align}
  \label{eq:1}
  W' \subseteq R_D(u) \text{\quad or \quad }  W' \subseteq R_D(v)\,\,.\tag{*}
\end{align}
\vspace{-4ex}

\noindent Again it holds that $W' \subseteq R_D(x)$ if and only if~$W \subseteq
R_{D'}(x)$ for~$x \in \{u,v\}$.

Given a planar straight-line drawing of a graph~$G$, we say that a
chord~$e$ is \emph{good} if its adjacent triangles or 4-cliques
satisfy condition~\eqref{eq:1}.  This notion gives us a more local
criterion to decide whether a given planar straight-line drawing is an
outside-obstacle representation.

\begin{lemma}
  \label{lem:nonlocal-crossing}
  Let~$G$ be a biconnected inner-chordal plane graph and let~$\Gamma$
  be a planar straight-line drawing of~$G$.  Then~$\Gamma$ is a
  (plane) outside-obstacle representation if and only if each chord is
  good.
\end{lemma}

\begin{proof}
  The condition that each chord is good is necessary.  For sufficiency
  we show that, in a drawing where each chord is good, every non-edge
  intersects the outer face.

  \begin{figure}[tb]
    \centering
    \includegraphics{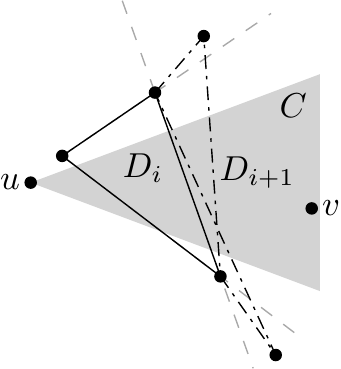}
    \caption{Illustration of the proof of
      Lemma~\ref{lem:nonlocal-crossing}.}
    \label{fig:non-local-crossing}
  \vspace{-2ex}
  \end{figure}

  Suppose for the sake of contradiction that~$u$ and~$v$ are two
  non-adjacent vertices of~$G$ such that the segment~$uv$ does
  \emph{not} intersect the outer face.  Then there is a minimal series
  $D_1, \dotsc, D_n$ of adjacent triangles of $G$ such that the
  segment~$uv$ is completely contained in the union of these
  triangles.  Clearly,~$n\ge2$ and, without loss of generality,~$u \in
  V(D_1)$ and~$v \in V(D_n)$.  We consider the subdrawing induced
  by~$D_1,\dots, D_n$.  Since~$uv$ intersects the triangle~$D_1$, it
  intersects the edge of~$D_1$ opposite of~$u$, which implies that~$v$
  is contained inside the cone~$C$ defined by~$D_1$ with base~$u$.  We
  show that this contradicts statement $2$.  If triangle~$D_i$ for~$1
  \le i \le n-1$ has the following properties: (i) its points are
  outside (or on the boundary) of~$C$, (ii) the edge shared by~$D_i$
  and~$D_{i+1}$ cuts across the cone~$C$, and (iii) the line defined
  by the two points of~$D_i$ that lie on the same side of~$C$ slope
  away from~$C$ in the direction towards~$v$, then the very same
  properties hold for~$D_{i+1}$ due to the chords being good; see
  Fig.~\ref{fig:non-local-crossing}.  By definition the property holds
  for~$D_1$, and hence it also holds for~$D_n$.  But this implies that
  the tip of~$D_n$, which is~$v$, must be placed outside of~$C$,
  contradicting the assumption.
\end{proof}

Unfortunately, it is not always possible to place the vertices inside
the regions such that all chords become good as this placement may
require crossings.

\boldparagraph{Chord Orientations and Outside-Obstacle Representations.}
Next, we introduce a certain type of orientations of the chords of
biconnected inner-chordal graphs.  Let~$G$ be a biconnected
inner-chordal graph and let~$\Gamma$ be a plane outside-obstacle
representation of~$G$.  Let~$e=\{u,v\}$ be a chord of~$G$, which
exists, unless~$G$ is~$K_3$ or~$K_4$.  Let~$D$ and~$D'$ denote the two
triangles incident to~$e$, and let~$w$ and~$w'$ denote the tips of~$D$
and~$D'$, respectively.  Due to Lemma~\ref{lem:nonlocal-crossing},
each chord satisfies condition~\eqref{eq:1}.  Hence, either~$w \in
R_{D'}(u)$ and~$w' \in R_{D}(u)$ or~$w \in R_{D'}(v)$ and~$w' \in
R_{D}(v)$.  We direct the chord~$e$ towards~$u$ in the former case and
towards~$v$ in the latter case.  In this way, we obtain an orientation
of the chords of~$G$.  Note that outer edges and inner edges of
4-cliques remain undirected.  The following lemma shows two crucial
properties of such an orientation.

\begin{lemma}
  \label{lem:chord-orientation}
  Let~$G$ be a biconnected inner-chordal graph with plane
  outside-obstacle representation~$\Gamma$.  The chord orientation
  determined by~$\Gamma$ satisfies the following properties.

  \begin{compactenum}[(i)]
  \item Each vertex has in-degree at most~2.
  \item If vertex~$v$ has in-degree~2, then its two incoming edges share a face.
  \end{compactenum}
\end{lemma}

\begin{proof}
  Consider an orientation according to~$\Gamma$.  Let~$e=(u,v)$ be a
  directed chord with incident triangles~$D$ and~$D'$, whose tips with
  respect to the base~$e$ are~$w$ and~$w'$, respectively.  Due to the
  direction of~$e$, we have that~$w \in R_{D'}(v)$ and~$w' \in
  R_D(v)$.  It is readily seen, e.g., in Fig.~\ref{fig:attachment},
  that the two angles at~$v$ incident to~$e$ sum up to more
  than~$\pi$.

  Let~$e_1,\dots,e_k$ be chords that are directed towards~$v$.
  Without loss of generality assume that these chords are numbered in
  the order of counterclockwise occurrence around~$v$, starting from
  the outer face.  Since the angles at~$v$ incident to each of these
  edges sum up to more than~$\pi$, it follows that some of these
  angles must coincide.  Due to the ordering, it follows that the
  angle at~$v$ right of~$e_i$ (with respect to the orientation
  towards~$v$) coincides with the left angle of~$e_{i+1}$
  for~$i=1,\dots,k-1$.  By planarity and since~$v$ is an outer vertex,
  no other angles may coincide.  For~$i=1,\dots,k$, let~$\alpha_i$
  denote the angle left of~$e_i$ and let~$\alpha_{k+1}$ denote the
  angle right of~$e_k$.  By the above observation, we
  have~$\alpha_i+\alpha_{i+1} > \pi$ for~$i=1,\dots,k$.

  For~$k\ge 3$, the sum of inner angles incident to~$v$ is at
  least~$\alpha_1 + \alpha_2 + \alpha_3+ \alpha_4 > 2\pi$; a
  contradiction.  For~$k=2$ the shared angle~$\alpha_2$ implies
  property~(ii).
\end{proof}

By virtue of Lemma~\ref{lem:chord-orientation}, we call any
orientation of the chords of a biconnected inner-chordal graph that
satisfies the properties~(i) and~(ii) an \emph{outside-obstacle
  orientation}.  

Lemma~\ref{lem:chord-orientation} finally allows us to give a concise
argument why the graph from Fig.~\ref{fig:example-2} does not admit a
plane outside-obstacle representation.  We argue that it does not
admit an outside-obstacle orientation.  It follows from the conditions
of such an orientation that, for each 4-clique that is incident to
three chords, these chords must be oriented such that they form a
cycle.  Consider the middle 4-clique in Fig.~\ref{fig:example-2}.  If
we orient it clockwise, then the lower left edge may not have
additional incoming chords, which prevents us from orienting the
chords of the left 4-clique as a cycle.  Symmetrically, choosing a
counterclockwise orientation for the middle 4-clique prevents correct
orientation of the right 4-clique.  The graph in
Fig.~\ref{fig:example-3}, however, does admit an outside-obstacle
orientation, which is indicated in the figure.

Our next goal is to prove that the existence of an outside-obstacle
orientation is equivalent to the existence of a plane outside-obstacle
representation.  In particular, this shows our claim that the graph in
Fig.~\ref{fig:example-3} indeed admits a plane outside-obstacle
representation, e.g., the one shown in Fig.~\ref{fig:example-oor-2}.

\begin{theorem}
  \label{thm:orientation-representation}
  Let~$G$ be a biconnected inner-chordal plane graph.  Then~$G$ admits
  a plane outside-obstacle representation if and only if it admits an
  outside-obstacle orientation.
\end{theorem}

\begin{proof}
  The ``only if''-part holds due to Lemma~\ref{lem:chord-orientation}.
  Let~$G$ be a biconnected inner-chordal graph with an
  outside-obstacle orientation and let~$T(G)$ be its construction
  tree.  We construct a plane outside-obstacle representation of~$G$.

  For a subtree~$T' \subseteq T$, we denote by~$G(T')$ the subgraph
  of~$G$ corresponding to~$T'$.  Note that~$G(T) = G$.
  Let~$\tau_1,\dots,\tau_k$ denote the nodes of~$T$ in breadth-first
  order starting at an arbitrary node~$\tau_1$.  For~$j=1,\dots,k$,
  let~$T_j$ be the subtree of~$T$ consisting of
  nodes~$\tau_1,\dots,\tau_j$, and let~$G_j = G(T_j)$ the
  corresponding subgraph of~$G$.  We inductively construct a sequence
  of plane outside-obstacle representations~$\Gamma_1,\dots,\Gamma_k$
  of~$G_1,\dots, G_k$.  Then~$\Gamma_k$ is the desired plane
  outside-obstacle representation of~$G=G_k$.

  Consider the orientation of~$G_i$ induced by~$G$ (note that some
  edges remain undirected).  We call a directed edge \emph{active} if
  it is incident to the outer face of~$G_i$ and \emph{inactive}
  otherwise.  An outer vertex~$v$ is \emph{active} if it is the target
  of an active edge.  It is \emph{inactive} otherwise.  The
  \emph{inactive degree} of~$v$ in~$G_i$ is the number of inactive
  edges with target~$v$.

  Throughout steps $i=1,\dots,k$, we maintain the following
  properties:
  \begin{compactenum}[(i)]
  \item The outer angle of vertices with inactive degree~0 is convex.
  \item For an active vertex~$v$ with inactive degree~1, removing the
    unique active in-edge incident to~$v$ results in a convex outer
    angle.
  \end{compactenum}
  For~$G_1$ any plane outside-obstacle representation~$\Gamma_1$ satisfies
  these properties.  We now show how to proceed from~$G_i$
  to~$G_{i+1}$.  Let~$e=(u,v)$ be the directed chord determined by
  adding~$\tau_{i+1}$ to~$T(G_i)$, let~$D$ be the inner triangle
  of~$G_i$ bounded by~$e$ and let~$e'=(u',v)$ denote the other edge
  of~$D$ incident to~$v$.

  We aim to place the vertices in~$V(G_{i+1}) \setminus V(G_i)$ inside
  the region~$R_D(v)$, which is consistent with the orientation
  of~$e$.  We first show that this is possible without creating
  intersections.  If~$v$ has inactive degree~0, then~$v$ is convex,
  and hence the intersection of~$R_D(v)$ with a suitably
  small~$\varepsilon$-ball around~$v$ is disjoint from any vertices
  and edges of~$\Gamma_i$.  Similarly, if~$v$ is active but has
  inactive degree~1, then, after removing~$(u,v)$,~$v$ is convex by
  property~(ii).  In this case the subcone of~$R_D(v)$ defined by~$e'$
  and the other outer edge incident to~$v$ intersected with a suitably
  small~$\varepsilon$-ball is empty; see Fig.~\ref{fig:construction-1}.

  \begin{figure}[tb]
    \centering
    \begin{subfigure}[b]{.4\textwidth}
      \centering
      \includegraphics[page=1]{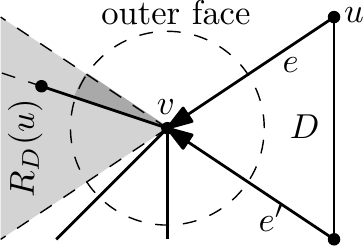}
      \caption{}\label{fig:construction-1}
   \end{subfigure}\hfil
   \begin{subfigure}[b]{.3\textwidth}
      \centering
      \includegraphics[page=3]{fig/construction}
      \caption{}\label{fig:construction-2}
   \end{subfigure}\hfil
  \begin{subfigure}[b]{.3\textwidth}
      \centering
      \includegraphics[page=2]{fig/construction}
      \caption{}\label{fig:construction-3}
   \end{subfigure}\hfil
   \caption{Construction of a plane outside-obstacle representation. (a)
     Attaching at a vertex~$v$ with inactive degree~1. The region
     where the new points are placed is shaded in dark gray. (b), (c)
     show that adding an edge preserves the properties required by
     the construction.  The shaded region is the inner part of the
     drawing after removing the active edge; the outer angle at~$v$ is
     convex.}
    \label{fig:construction}
    \vspace{-2ex}
  \end{figure}

  We show that, by placing the new vertices in these regions suitably
  close to~$v$, properties~(i) and~(ii) can be established for the
  resulting plane outside-obstacle representation~$\Gamma_{i+1}$.  First
  observe that placing the new vertices close enough to~$v$ avoids
  touching or crossing any vertices and edges of~$G_i$, i.e.,
  $\Gamma_{i+1}$ is plane.  Moreover, the addition only changes angles
  at~$u$ and~$v$, and hence all other vertices satisfy properties~(i)
  and~(ii) by virtue of the induction hypothesis.

  Consider vertex~$v$.  In~$G_{i+1}$ it has inactive degree at least~1
  since~$(u,v)$ is an inner edge.  If the inactive degree of~$v$ is~2,
  there is nothing to prove as~$v$ must be inactive, since
  outside-obstacle orientations have in-degree at most~2.  Hence
  assume that~$v$ has inactive degree~1 and it is active.  The
  properties of outside-obstacle orientations imply that there is a
  unique active edge directed towards~$v$ in~$G_{i+1}$ and it must be
  a neighbor of~$e$.  This edge is either the edge~$e'$ or the newly
  added outer edge~$e''$ incident to~$v$.

  If~$e'$ is the incoming active edge at~$v$, the outer angle
  at~$v$ was convex in~$\Gamma_i$, and hence any point in~$R_D(v)$
  results in an outer angle of less than~$\pi$ after removing~$e'$;
  see Fig.~\ref{fig:construction-2}.
  If~$e''$ is the incoming active edge at~$v$, the outer angle
  at~$v$ in~$\Gamma_{i+1}$ after removing~$e''$ is the outer angle
  of~$v$ in~$\Gamma_i$, which is convex by the induction hypothesis;
  see Fig.~\ref{fig:construction-3}.

  In all cases vertex~$v$ satisfies properties~(i) and~(ii).  We show
  that, by positioning the new vertices close enough to~$v$, we can
  also satisfy properties~(i) and~(ii) for $u$.  First note that the
  inactive degree of~$u$ does not change.  If the inactive degree
  of~$u$ is~2, there is nothing to prove.  If the inactive degree
  of~$u$ is~0 or~1, by placing the new vertices close to the line
  through~$u$ and~$v$, the angle between~$e$ and the new outer edge
  incident to~$u$ can be made arbitrarily small.  Thus, if~$u$ was
  convex in~$\Gamma_i$, it remains so in~$\Gamma_{i+1}$.  And, by the
  same argument, if~$u$ was convex in~$\Gamma_i$ after removing the
  active edge incident to~$u$, it remains so in~$\Gamma_{i+1}$.
  Hence~$\Gamma_{i+1}$ satisfies the induction hypothesis.
\end{proof}

\section{Characterization and Decision Algorithm}
\label{sec:main}

In this section, we prove characterizations of graphs that admit a
plane outside-obstacle representation and we present a linear-time
algorithm that decides whether a given graph admits a plane
outside-obstacle representation.

\boldparagraph{Characterization of Outerplanar Graphs.}
For biconnected outerplanar graphs
Theorem~\ref{thm:orientation-representation} immediately implies a
complete characterization of the graphs that admit a plane
outside-obstacle representation.

\begin{theorem}
  \label{thm:bico-outer}
  A biconnected outerplanar graph admits a plane outside-obstacle
  representation if and only if it is chordal.
\end{theorem}

\begin{proof}
  Being chordal is a necessary condition due to
  Lemma~\ref{lem:plane-oor-chordal}.  Conversely, if an outerplanar
  graph is chordal, it is obviously inner-chordal.  We show that every
  biconnected inner-chordal outerplane graph admits an outside
  obstacle orientation.

  Recall that a biconnected outerplanar graphs contains a vertex of
  degree at most~2.  We iteratively construct an orientation by
  directing the incident edges of a vertex with degree at most~2
  towards it and removing it from the graph.  In this way, we obtain
  an orientation with the properties that each vertex has in-degree at
  most~2, and moreover, if a vertex has in-degree~2, then the two
  incoming edges share an inner face.  Undoing the orientations of the
  outer edges, we obtain an outside-obstacle orientation.  Now the
  claim follows from Theorem~\ref{thm:orientation-representation}.
\end{proof}

This result can easily be strengthened in two ways.  First, if the
outerplanar graph is not biconnected but chordal, then it can easily
be augmented such that it becomes biconnected but remains
(inner-)chordal and outerplanar and hence satisfies the conditions of
Theorem~\ref{thm:bico-outer}, yielding a plane outside-obstacle
representation of the augmented graph.  By iteratively removing
augmentation edges that are incident to the outer face we obtain a plane
outside-obstacle representation of the original graph.

\begin{corollary}
  \label{cor:outer}
  An outerplanar graph admits a plane outside-obstacle representation
  if and only if it is chordal.
\end{corollary}

Another observation is that the construction of the orientation in the
proof of Theorem~\ref{thm:bico-outer} essentially consists of a
bottom-up traversal of the construction tree of the graph with respect
to the root node, which is removed last.  It is then readily seen that
we can also remove~$K_4$-nodes that are leaves, provided they have
degree at most~2 in~$T(G)$.  A $K_4$ with degree~3 requires that its
chords are oriented to form a cycle, which cannot be ensured by the
construction.  It can, however, always be achieved it the $K_4$ of
degree~3 is the root of the tree.  We thus have the following
corollary.

\begin{corollary}
  \label{cor:bico-onek4deg3}
  Every biconnected inner-chordal graph that contains at most one
  $K_4$ for which all outer edges are chords admits a plane
  outside-obstacle representation.
\end{corollary}

Note that an augmentation as in the proof of
Corollary~\ref{cor:bico-onek4deg3} may increase the number of
$K_4$-nodes with degree~3.  Hence the result does not extend to
non-biconnected graphs.

\boldparagraph{Decision Algorithm for General Graphs.}
Next, we devise a linear-time algorithm to decide whether a
biconnected graph admits a plane outside-obstacle representation.  Of
course it is not difficult to test whether a graph is inner-chordal
and plane in linear time, and we assume in the following that our
input graph has these properties.

Due to Theorem~\ref{thm:orientation-representation}, deciding the
existence of a plane outside-obstacle representation is equivalent to
deciding the existence of an outside-obstacle orientation.  To test
whether a biconnected inner-chordal plane graph~$G$ admits an outside
obstacle orientation, we use dynamic programming on its construction
tree~$T(G)$, rooted at an arbitrary node.  

For each node~$\tau$ with parent edge~$\{u,v\}$ with orientation~$uv$
and binary flags~$d_{\tau,u}$ and~$d_{\tau,v}$, we are
  interested whether the subtree of~$T(G)$ with root~$\tau$ admits an
  outside-obstacle orientation such that
\begin{compactenum}
\item $\{u,v\}$ is oriented as $uv$,
\item $u$ has incoming edges if and only if~$d_{\tau,u} = 1$, and
\item $v$ has incoming edges distinct from~$uv$ if and only
  if~$d_{\tau,v} = 1$.
\end{compactenum}

We store this information in a 4-dimensional
table~$T[\tau,e,d_{\tau,u},d_{\tau,v}]$ of boolean variables.
Note that, for each node~$\tau$, table~$T$ contains only~$2^3 = O(1)$
entries.  We now show how to fill the entries of this table in linear
time.  Initially, we set all entries to \emph{false}.

For a leaf node~$\tau$ with parent edge~$\{u,v\}$, we
set~$T[\tau,uv,0,0] = T[\tau,vu,0,0] =$ \emph{true}, which models the
fact that we can choose any orientation of~$\{u,v\}$ and neither~$u$
nor~$v$ has incoming edges distinct from~$\{u,v\}$ in the subtree
consisting only of the leaf.  Let~$\tau$ be a node with
children~$\tau'$ and~$\tau''$ and corresponding chords~$\{u,w\},
\{v,w\}$ that connect them to~$\tau$.  We can easily check whether the
entries can be combined to an entry of~$\tau$.  Namely, try both
possible orientations of~$\{u,v\}$ and use the orientations
of~$\{u,w\}$ and~$\{v,w\}$ determined by the entries of the children
and the flags~$d_{\tau',u}$, $d_{\tau',w}$, $d_{\tau'',v}$,
and~$d_{\tau'',w}$ of the children to check that~$u,v$ and~$w$ satisfy
the constraints of the orientation.  If this is the case, we can
easily compute the two flags~$d_{\tau,u}$ and~$d_{\tau,v}$ from the
orientations of~$uw$, $vw$ and the flags~$d_{\tau',u}$
and~$d_{\tau'',v}$.  A simple induction shows that, in this way, we
set exactly the correct entries~$T[\tau,\cdot,\cdot,\cdot]$ to
\emph{true}.

Combining two entries takes~$O(1)$ time.  Since there are only~$2^3 =
O(1)$ entries per node, we can compute all entries of a node~$\tau$
from all combinations of entries of its at most two children in~$O(1)$
time.  Since there are~$O(n)$ nodes, the overall algorithm runs
in~$O(n)$ time.  At the root we may have to combine up to three
children, but the checks remain essentially the same.  Thus, the
overall algorithm runs in~$O(n)$ time.

\begin{theorem}
  There is a linear-time algorithm that decides whether a given
  biconnected graph admits a plane outside-obstacle representation.
\end{theorem}

\section{Conclusion}
\label{sec:conclusion}

Inspired by obstacle representations introduced by Alpert et
al.~\cite{akl-ong-10}, we studied plane outside-obstacle
representations of graphs.  We characterized the biconnected graphs
that admit such a representation as the inner-chordal graphs that
admit a certain type of orientation of their chords.  Based on this,
we gave a combinatorial proof that every chordal outerplanar graph
admits a plane outside-obstacle representation.  We further derived a
linear-time algorithm for deciding whether a given biconnected graph
admits a plane outside-obstacle representation.

Our main open question are the following.  Can our characterization
and testing algorithm can be extended to general (inner-chordal)
graphs that are not necessarily biconnected?  Which graphs admit a
plane representation with a single obstacle?

\smallskip

\noindent\textbf{Acknowledgments}
Part of this work has been done while Alexander Koch participated in
the academic exchange program of T\=ohoku University and KIT.  AK
thanks Prof.~Dorothea Wagner and Prof.~Takeshi Tokuyama for their
support and Prof.~Yota \=Otachi from JAIST for helpful comments on the topic.

\bibliographystyle{abbrv}
\bibliography{vis_graphs}

\end{document}